\documentclass[a4paper,12pt,table]{article}
\usepackage{amssymb,amsmath,latexsym,color,xcolor,amsthm,authblk,mathpazo,palatino,bbding,setspace,datetime,breakcites,hyphenat,microtype,diagbox,tikz,subcaption,graphicx,titlesec}
\usepackage[round,authoryear]{natbib}
\usepackage{hyperref,theoremref}
\usepackage{cleveref}
\usepackage{xspace}
\usepackage{comment}
\usepackage{xifthen}
\usepackage[titletoc,title]{appendix}
\usepackage[shortlabels]{enumitem}
\usepackage{graphicx}
\usepackage{float}
\usepackage{amssymb}
\restylefloat{table}
\everymath{\displaystyle}
\usepackage{pstricks}

\theoremstyle{plain}

\definecolor{armygreen}{rgb}{0.29, 0.8, 0.13}
\definecolor{auburn}{rgb}{0.43, 0.21, 0.1}
\definecolor{burgundy}{rgb}{0.5, 0.0, 0.13}
\definecolor{medium red}{rgb}{.490,.298,.337}
\definecolor{dark red}{rgb}{.235,.141,.161}

\hypersetup{
	colorlinks = true,
	linkcolor = {burgundy},
	citecolor = {burgundy}, 		
	linkbordercolor = {white},
}

\let\OLDthebibliography\thebibliography
\renewcommand\thebibliography[1]{
	\OLDthebibliography{#1}
	\setlength{\parskip}{0pt}
	\setlength{\itemsep}{0pt plus 0.1ex}
}

\DeclareFontFamily{U}{mathx}{\hyphenchar\font45}
\DeclareFontShape{U}{mathx}{m}{n}{<-> mathx10}{}
\DeclareSymbolFont{mathx}{U}{mathx}{m}{n}
\DeclareMathAccent{\widebar}{0}{mathx}{"73}

\titleformat{\section}[block]{\normalfont\scshape\large\filcenter}{\thesection}{1em}{}
\titleformat{\subsection}{\normalfont\scshape\large}{\thesubsection}{1em}{}
\titleformat{\subsubsection}{\normalfont\scshape\large}{\thesubsubsection}{1em}{}

\newtheorem{theorem}{Theorem}
\newtheorem*{theorem*}{Theorem}

\newtheorem{claim}{Claim}

\theoremstyle{definition}
\newtheorem{definition}{Definition}

\theoremstyle{remark}

\makeatletter
\newcommand*\bigcdot{\mathpalette\bigcdot@{.5}}
\newcommand*\bigcdot@[2]{\mathbin{\vcenter{\hbox{\scalebox{#2}{$\m@th#1\bullet$}}}}}
\makeatother
\newcommand{\objects}{\ensuremath{X}\xspace}
\newcommand{\randoma}{\ensuremath{A}\xspace}
\newcommand{\randomaof}[2]{
    \ifthenelse{\isempty{#1}}
    {
        \ifthenelse{\isempty{#2}}
        {\ensuremath{\randoma_{i\boldsymbol{\bigcdot}}}\xspace}
        {\ensuremath{\randoma_{i#2}}\xspace}
    }
    {
        \ifthenelse{\isempty{#2}}
        {\ensuremath{\randoma_{#1\boldsymbol{\bigcdot}}}\xspace}
        {\ensuremath{\randoma_{#1#2}}\xspace}
    }
}

\newcommand{\rendow}[2]{
    \ifthenelse{\isempty{#1}}
        {
            \ifthenelse{\isempty{#2}}
            {\ensuremath{E_{i\boldsymbol{\bigcdot}}}\xspace}
            {\ensuremath{E_{i#2}}\xspace}
        }
        {
            \ifthenelse{\isempty{#2}}
            {\ensuremath{E_{#1\boldsymbol{\bigcdot}}}\xspace}
            {\ensuremath{E_{#1#2}}\xspace}
        }
}
\newcommand{\dendow}[1]{
    \ifthenelse{\isempty{#1}}
    {\ensuremath{E_{i}}\xspace}
    {\ensuremath{E_{#1}}\xspace}
}
\newcommand{\rrule}{\ensuremath{\varphi}\xspace}
\newcommand{\rruleof}[3]{
    \ifthenelse{\isempty{#1}}
    {
        \ifthenelse{\isempty{#2}}
        {\ensuremath{\rrule_{i\boldsymbol{\bigcdot}}(#3)}\xspace}
        {\ensuremath{\rrule_{i#2}(#3)}\xspace}
    }
    {
        \ifthenelse{\isempty{#2}}
        {\ensuremath{\rrule_{#1\boldsymbol{\bigcdot}}(#3)}\xspace}
        {\ensuremath{\rrule_{#1#2}(#3)}\xspace}
    }
}
\newcommand{\without}[1]{
    \ifthenelse{\isempty{#1}}
    {\ensuremath{P_{-i}}}
    {\ensuremath{P_{-#1}}}
}
\newcommand{\prefof}[3]{
    \ifthenelse{\isempty{#3}}
    {
        \ifthenelse{\isempty{#1}}
            {
                \ifthenelse{\isempty{#2}}
                {\ensuremath{P_i(1)}\xspace}
                {\ensuremath{P_i(#2)}\xspace}
            }
            {
                \ifthenelse{\isempty{#2}}
                {\ensuremath{P_{#1}(1)}\xspace}
                {\ensuremath{P_{#1}(#2)}\xspace}
            }
    }
    {
        \ifthenelse{\isempty{#1}}
            {
                \ifthenelse{\isempty{#2}}
                {\ensuremath{P_i^{#3}(1)}\xspace}
                {\ensuremath{P_i^{#3}(#2)}\xspace}
            }
            {
                \ifthenelse{\isempty{#2}}
                {\ensuremath{P_{#1}^{#3}(1)}\xspace}
                {\ensuremath{P_{#1}^{#3}(#2)}\xspace}
            }
    }
}

\newcommand{\suml}{\sum\limits_}

\newdateformat{monthyeardate}{\monthname[\THEMONTH], \THEYEAR}

\newcommand{\incoming}[1]{\ensuremath{\downarrow\!\!(#1)}}
\newcommand{\incomingm}[1]{\ensuremath{\downarrow(#1)}}
\newcommand{\outgoing}[1]{\ensuremath{\uparrow\!\!(#1)}}
\newcommand{\outgoingm}[1]{\ensuremath{\uparrow(#1)}}

\title{\textsc{On Probabilistic Assignment Rules}}
\author[]{Gogulapati Sreedurga\thanks{University of Edinburgh; Email: sgogulap@ed.ac.uk}}
\author[]{Yadati Narahari\thanks{Indian Institute of Science Bangalore; Email: narahari@iisc.ac.in}}
\author[]{Souvik Roy\thanks{Indian Statistical Institute, Kolkata; Email: gametheory.souvik@gmail.com}} 
\author[]{Soumyarup Sadhukhan\thanks{Indian Institute of Technology Kanpur; Email: soumyarups@iitk.ac.in}}
\affil[]{}
\date{\monthyeardate\today}

\begin{document}
	
	\maketitle
	
	\begin{abstract}\singlespacing
    We study the classical assignment problem with initial endowments in a probabilistic framework. In this setting, each agent initially owns an object and has strict preferences over the entire set of objects, and the goal is to reassign objects in a way that satisfies desirable properties such as strategy-proofness, Pareto efficiency, and individual rationality. While the celebrated result by \cite{ma1994strategy} shows that the Top Trading Cycles (TTC) rule is the unique deterministic rule satisfying these properties, similar positive results are scarce in the probabilistic domain. We extend Ma’s result in the probabilistic setting, and as desirable properties, consider SD-efficiency, SD-individual rationality, and a weaker notion of SD-strategy-proofness—SD-top-strategy-proofness—which only requires agents to have no incentive to misreport if doing so increases the probability of receiving their top-ranked object. We show that under deterministic endowments, a probabilistic rule is SD-efficient, SD-individually rational, and SD-top-strategy-proof if and only if it coincides with the TTC rule. Our result highlights a positive possibility in the face of earlier impossibility results for fractional endowments (\cite{athanassoglou2011house}) and provides a first step toward reconciling desirable properties in probabilistic assignments with endowments.
		
		\noindent 
				
		\vspace{4mm}
		
        
\noindent MSC2010 Subject Classification: 91B14, 91B03, 91B68
		
		\vspace{4mm}
		
	\end{abstract}
	
	\newpage

	\maketitle
	
{
	\def\OldComma{,}
	\catcode`\,=13
	\def,{%
		\ifmmode%
		\OldComma\discretionary{}{}{}%
		\else%
		\OldComma%
		\fi%
	}%


\section{Introduction}
\subsection{Problem Description  and Related Literature}

We consider the classical assignment problem with initial endowments in the probabilistic setup. In an assignment problem, a finite set of objects has to be allocated among a finite set of agents (with the same cardinality)  based on their (strict) preferences over the objects. Moreover, the agents have endowments over the objects. Practical applications of this model can be thought of as house allocation among existing tenants, course allocation among the faculty members who have endowments as the courses they taught last time, and many others. A deterministic assignment rule (henceforth, deterministic rule) is a function from the set of reported preferences to the set of assignments, where an assignment is a 1-1 mapping between the set of agents and the set of objects. Some desirable properties of deterministic rules are strategy-proofness, Pareto-efficiency, individual rationality, pair-efficiency, etc.; a deterministic rule is strategy-proof if no agent has an incentive to misreport her preference to get a better outcome. Pareto-efficiency implies there is no way the outcome of the deterministic rule can be shuffled among the agents and make everybody weakly better off, whereas somebody is strictly better off. Finally, a deterministic rule is individually rational if every agent is guaranteed to get a better outcome than her initial endowment according to her preference. 


\cite{ma1994strategy} shows that a deterministic rule satisfies strategy-proofness, Pareto-efficiency, and individual rationality if and only if it is the TTC rule, a deterministic rule first introduced in \cite{shapley1974cores} in the context of the object reallocation problem. \cite{svensson1994queue}, \cite{anno2015short}, and \cite{sethuraman2016alternative} 
provided shorter proofs of this result. Recently, \cite{ekici2024pair} shows that the same result holds even if we replace Pareto-efficiency with a much weaker condition called pair-efficiency. A rule is pair-efficient if two agents cannot exchange their objects assigned by the rule and both become better off. 

In the probabilistic setup, instead of a deterministic rule, a probabilistic assignment rule (henceforth, probabilistic rule)  is considered, which assigns a bi-stochastic matrix at every instance of reported preferences of the agents. Here, the order of a bi-stochastic matrix is the common cardinality of the agent set and the object set, and each row of the matrix denotes the probabilities of assigning the objects to a particular agent, whereas the columns represent the probabilities of assigning a specific object to different agents. In the mechanism design literature, it is well-known that probabilistic rules are better in terms of fairness consideration when compared to their deterministic counterpart. Moreover, considering probabilistic rules broadens the scope of the designer as the class of rules becomes much larger in the case of probabilistic rules. For probabilistic rules, the corresponding desirable properties like strategy-proofness, Pareto-efficiency, individual rationality, and pair-efficiency can be defined using the first-order stochastic dominance. There have been quite a few works in the probabilistic setup of the assignment problem with initial endowments. \cite{athanassoglou2011house} shows that when initial endowments are probabilistic (also known as fractional endowments), SD-strategy-proofness, SD-efficiency, and SD-individual rationality are together incompatible when there are at least four objects. In a later paper, \cite{aziz2015generalizing} strengthens the result by weakening SD-strategy-proofness with weak-SD-strategy-proofness. In both their proofs, they started from some specific fractional endowments and arrived at an impossibility. To the best of our knowledge, there is not much known in the probabilistic setup when the feasible set of initial endowments is restricted.

\subsection{Motivation and Contribution of the Paper}\label{se_1.2}

We extend the main result in \cite{ma1994strategy} in the probabilistic setup with a deterministic endowment. As explained earlier, the probabilistic setup is more appealing than the deterministic one because of the fairness criteria, and it broadens the scope of the designer.  
While Pareto-efficiency and individual rationality have straightforward extensions in the probabilistic setup using first-order stochastic dominance, many extensions are possible of the notion of strategy-proofness (see \cite{sen2011gibbard, aziz2015generalizing, chun2020upper}). The most used one in the literature (\cite{gibbard1977manipulation}) assumes that an agent will manipulate if by misreporting she can be better off for at least one upper contour set. In other words, for strategy-proofness, it requires that truth-telling ensures that for every upper contour set, there is no misreport that can yield a bigger probability than the one obtained by reporting truthfully. No doubt, this is quite a strong condition and may not be appropriate in a practical scenario, as in reality, agents might not care for all upper contour sets. Keeping this point in mind, we consider a weakening of strategy-proofness that we call as SD-top-strategy-proofness.
 Top-strategy-proofness assumes the agents only care about manipulation if they can increase the probability of their top alternative.
The deterministic version of top-strategy-proofness implies agents manipulate only if, by misreporting, they can ensure their top alternative as the outcome. Thus, top-strategy-proofness is the minimal form of strategy-proofness one may think of. In Theorem \ref{the: detend}, we prove that a probabilistic rule is Pareto-efficient, individually rational, and top-strategy-proof if and only if it is the TTC rule. Our result generalizes \cite{ma1994strategy}'s result in the deterministic setup as we work with SD-top-strategy-proofness.

Our result can also be seen as the first step in the direction of evading the negative result of \cite{athanassoglou2011house}. As mentioned earlier, \cite{athanassoglou2011house} show that in the probabilistic setup with fractional endowment, strategy-proofness, Pareto efficiency, and individual rationality are incompatible. To circumvent this negative result, one possible way could be to restrict the set of fractional endowments, i.e., considering only a subset of the set of all bi-stochastic matrices as feasible initial endowments, and see what happens then. In this paper, we take the first step towards this by showing that if we consider the set of permutation matrices for initial endowments, TTC is the only rule that survives, satisfying the three properties.


Our proof uses an inductive argument on the number of agents. To show that at a profile $P_N$, the probabilistic assignment will be the TTC outcome, we construct another profile $P_N'$ by modifying the preferences in $P_N$. In $P_N'$, we use individual rationality and Pareto efficiency to conclude that the outcome will match with TTC. We then use this result to show that the same holds for $P_N$ by changing the preferences of the agents one by one, and using top-strategy-proofness and individual rationality at every step. We believe this technique will be useful to further generalize this result in the presence of more relaxed versions of the three above-mentioned properties.  

\subsection{Other Related Literature}
In the probabilistic assignment problem, one of the first papers is \cite{bogomolnaia2001new}. They introduce a new probabilistic rule called Probabilistic Serial (PS rule) and show that it satisfies some nice properties over the Random Priority rule. They also show an important impossibility result that there is no probabilistic rule satisfying  SD-efficiency, equal treatment of equals and  SD-strategy-proofness. Equal treatment of equals ensures that two agents with the same preference get the same outcome. Later on, \cite{bogomolnaia2012probabilistic} characterize the PS rule in terms of  SD-efficiency, SD-no-envy, and bounded invariance, where SD-envy-free requires that every agent prefers their share over others. Bounded invariance is a weaker notion of SD-strategy-proofness.
\cite{chun2020upper} strengthen the impossibility result of \cite{bogomolnaia2001new} by weakening SD-strategy-proofness to upper-contour strategy-proofness, which only requires that if the upper-contour sets of some objects are the same in two preference relations, then the sum of probabilities assigned to the objects in the two upper-contour sets should be the same. \cite{yilmaz2009random} considers probabilistic assignment problems under weak preferences. Their main contribution is a recursive solution for the weak preference domain that satisfies individual rationality, ordinal efficiency and no justified-envy. No justified-envy views an assignment as unfair if an agent does not prefer his consumption to another agent's consumption, and the assignment obtained by swapping their consumptions respects the individual rationality requirement of the latter agent. Another paper that considers probabilistic allocations with a deterministic initial allocation is \cite{abdulkadirouglu1999house}. In their setting, there are existing tenants as well as new applicants. They introduce a deterministic rule called top trading cycles with fixed priority, which boils down to the TTC rule if there are only existing tenants. They also consider probabilistic allocations, which are a convex combination of these deterministic rules and show that these convex combination rules are strategy-proof, efficient, and individually rational.

We organize the paper in the remaining sections as follows: Section \ref{sec: prelims} introduces the model and the required definitions. Section \ref{sec: ttc} describes the TTC rule in an algorithmic way. Finally, our main result (Theorem \ref{the: detend}) is presented in Section \ref{sec: results}.


\section{Preliminaries}\label{sec: prelims}
Let $N=\{1,\ldots,n\}$ be a finite set of agents. Except otherwise mentioned, $n \geq 2$. Let $\objects = \{x_1,\ldots,x_n\}$ be a finite set of objects. A  reflexive, anti-symmetric, transitive, and complete binary relation (also called a linear order) on the set \objects is called a preference on \objects. We denote by $\mathcal{P}$ the set of all preferences on \objects. For $P \in \mathcal{P}$ and  $x,y \in \objects$, $xPy$ is interpreted as "$x$ is as good as (that is, weakly preferred to) $y$ according to $P$". Since $P$ is complete and antisymmetric, for distinct $x$ and $y$, we have either $xPy$ or $yPx$, and in such cases, $xPy$ implies $x$ is strictly preferred to $y$. For $P \in \mathcal{P}$ and $k \leq n$, by $P(k)$ we refer to  the $k$-th ranked object in \objects according to $P$, i.e., $P(k)=x$ if and only if $|\{y \in \objects\mid yP x \}|=k$. We use $r_i(x)$ to denote the rank of object $x$ in the preference of agent $i$, i.e., $r_i(x) = |\{y \in X \mid yP_ix\}|$. For $P \in \mathcal{P}$ and $x \in \objects$, the \textit{upper contour set} of $x$ at $P$, denoted by $U(x,P)$, is defined as the set of objects that are as good as $x$ in $P$, i.e., $U(x,P)=\{y \in \objects \mid yPx\}$.\footnote{Observe that $x\in U(x,P)$ by reflexivity.} An element  $P_N=(P_1,\ldots,P_n)$ of $\mathcal{P}^n$ is called a preference profile.

\subsection{Probabilistic assignments and their properties}

A probabilistic assignment \randoma is a bi-stochastic matrix of order $n$, i.e., \randoma is an $n \times n$ matrix in which every entry is in between $0$ and $1$ and every row and column sums are  $1$ (i.e., $0 \leq \randomaof{i}{j} \leq 1$, $\suml{j \in [m]}{\randomaof{i}{j}}=1$, and $\suml{i \in [n]}{\randomaof{i}{j}}=1$). We denote by $\mathcal{\randoma}$ the set of all probabilistic assignments. When each entry of a probabilistic assignment \randoma is either zero or one, then it is called a deterministic assignment. Note that a deterministic assignment is a permutation matrix.

The rows of a probabilistic assignment correspond to the agents, and the columns correspond to the objects. For an agent $i$ and object $x$, the value $\randomaof{i}{j}$ denotes the probability with which agent $i$ receives object $x_j$. Using standard matrix notations, we write: $\randomaof{i}{\boldsymbol{\bigcdot}} \in  \Delta(\objects)$ denotes the probability distribution in the $i^{\text{th}}$ row of \randoma.\footnote{Let $\Delta (S)$ denote the set of probability distributions on a set $S$.} Likewise, $\randomaof{\boldsymbol{\bigcdot}}{j} \in \Delta(N)$ denotes the probability distribution in the column corresponding to the object $x_j$. Further, with a slight abuse of the above notations, for any set $S \subseteq \objects$, let $\randomaof{i}{S}$ denote the total probability of the elements in $S$ in the $i^{\text{th}}$ row of \randoma. That is, $\randomaof{i}{S}  = \suml{j \in S}{\randomaof{i}{j}}$. Likewise, for any $S \subseteq N$, we define $\randomaof{S}{j}$ to be the total probability of the elements in $S$ in the $j^{\text{th}}$ column of \randoma.

In this paper, we assume that there is an initial deterministic assignment of the objects to the individuals. Let $E$ (a permutation matrix) denote that initial assignment, also referred to as \textbf{initial endowment}. For notational simplicity, we assume agent $i$ has the object $x_i$ as her initial assignment. This means essentially $E$ is the identity matrix of order $n$.

A \textbf{probabilistic assignment rule} (henceforth, a PR) is a function $\rrule: \mathcal{P}^n \to \mathcal{\randoma}$. The entry $\rruleof{i}{x_j}{P_N}$ represents the probability with which agent $i$ receives object $x_j$ at the profile $P_N$. A PR $\rrule$ is said to be \textbf{deterministic assignment rule} (henceforth, a DR and typically denoted by $f$) if for every $P_N \in \mathcal{P}^n$, $\rrule(P_N)$ is a deterministic assignment.

Given a preference $P_i \in \mathcal{P}$ of the agent $i$, we say that the agent $i$ \textbf{weakly prefers} a probability distribution $\lambda \in \Delta (\objects)$ over a distribution $\lambda' \in \Delta(\objects)$, if for every $x_j \in \objects$, we have $\lambda(U(x,P_i)) \geq \lambda'(U(x,P_i))$. Likewise, agent $i$ \textbf{strictly prefers} $\lambda$ over $\lambda'$ if it weakly prefers $\lambda$ over $\lambda'$ and there exists some $x \in \objects$ such that $\lambda(U(x,P_i)) > \lambda'(U(x,P_i))$. We write $\lambda \succeq_{P_i} \lambda'$ to indicate that $i$ weakly prefers $\lambda$ over $\lambda'$, and $\lambda \succ_{P_i} \lambda'$ to indicate that $i$ strongly prefers $\lambda$ over $\lambda'$.

We now introduce three desirable properties of probabilistic assignment rules.

\begin{definition}(SD-Pareto -efficiency)\label{def: pe} A probabilistic assignment \randoma SD-Pareto-dominates another probabilistic assignment $\randoma'$ at a profile $P_N$ if $$\randoma_{i\boldsymbol{\bigcdot}} \succeq_{P_i} \randoma'_{i\boldsymbol{\bigcdot}} \text{ for all } i\in N \text{ and } \randoma_{j\boldsymbol{\bigcdot}} \succ_{P_j} \randoma'_{j\boldsymbol{\bigcdot}} \text{ for some } j\in N.$$ A probabilistic assignment is said to be SD-Pareto-dominated if it is SD-Pareto-dominated by some other probabilistic assignment.  

A probabilistic assignment rule $\rrule$ is \textbf{SD-efficient} if for all $P_N\in \mathcal{P}^n$, $\rrule(P_N)$ is not SD-dominated at $P_N$. 	
\end{definition}


\begin{definition}(SD-Individual Rationality)\label{def: ir}
A probabilistic assignment rule $\rrule$ is \textbf{SD-individually rational} if for every $P_N \in \mathcal{P}^n$ and every agent $i \in N$, we have $$\rruleof{}{}{P_N} \succeq_{P_i} \rendow{}{} \;\;.$$
\end{definition}

\begin{definition}(SD-Strategy-proofness)\label{def: sp}
A probabilistic assignment rule $\rrule$ is \textbf{SD-strategy-proof} if for every $P_N \in \mathcal{P}^n$, every agent $i \in N$, and every $P_i' \in \mathcal{P}$, $$\rruleof{}{}{P_N} \succeq_{P_i} \rruleof{}{}{P_i',\without{}}.$$
\end{definition}

As mentioned in Section \ref{se_1.2}, below we formally define the weaker notion of SD-strategy-proofness that we work with in the results. 

\begin{definition}[SD-top-strategy-proofness]\label{def: topsp}
A probabilistic assignment rule $\rrule$ is \textbf{top-SD-strategy-proof} if for every $P_N \in \mathcal{P}^n$, every agent $i \in N$, and every $P_i' \in \mathcal{P}$, $$\rruleof{}{\prefof{}{}{}}{P_N} \geq \rruleof{}{\prefof{}{}{}}{P_i',\without{}}.$$
\end{definition}

All the above properties can be defined similarly for the deterministic assignments, and we refer to them as Pareto-efficiency, individual rationality, strategy-proofness, and top-strategy-proofness.  We now study a rule, called Top Trading Cycles, that starts with an initial endowment. Notably, the rule always results in a deterministic assignment.

\section{Top Trading Cycles (TTC) Rule}\label{sec: ttc} 

Now, we define the TTC rule. As said earlier, we assume agent $i$ is endowed with object $x_i$ initially. We use $\objects_S$ to denote the set of endowments of all the agents in $S \subseteq N$. We additionally need the following terminology. For a preference $P$ on \objects and a subset $Y$ of \objects, we write $P|_Y$ to refer to the restriction of $P$ to $Y$, that is, $P|_Y$ is the preference on $Y$  such that for all $x,y \in Y$, $xP|_Yy$ if and only if $xPy$. For a preference profile $P_N$, we write $P_N|_Y$ to refer to the profile in which all the preferences in $P_N$ are restricted to $Y$.


We now introduce a particular type of graph. Let $S \subseteq N$ and $P_i$ be a preference on $X_S$ for an agent $i \in S$. The graph $\mathcal{G}(P_S)$ is defined as follows: the  set of nodes is $S$ and there is an edge $(i,j)$ if and only if $\prefof{}{}{}=x_j$. Clearly, the graph $\mathcal{G}(P_N)$ will have at least one cycle. 

We are now ready to define the TTC rule. Consider a preference profile $P_N$. The TTC rule is a deterministic assignment rule whose step-by-step description at this profile is as follows:
\begin{enumerate}
    \item \label{nodes} Round 1: Let $\objects^1=\objects$,  $N^1=N$, and $P^1_{N^1}=P_N$. Consider the graph  $\mathcal{G}^1(P^1_{N^1})$. 
    Let $C=(i_1,\ldots, i_k,i_1)$ be the smallest cycle in the graph. Assign $\dendow{j+1}$ to agent $j$ for all $j \in \{1,\ldots,k\}$ (where $i_{k+1}=i_1$). 
    
    \item  \label{intialize}  Round 2: Let  ${N}^2$ be the set of remaining agents (that is, the agents who did not belong to any cycle in the first round). Let $\objects^2=X_{{N}^2}$ and  ${P}^2_{{N}^2}$ be the reduced preference profile of the agents in $N_2$ to the set of objects in $\objects^2$, that is, ${P}^2_i=P_i^1|_{\objects^2}$. Now, consider the graph $\mathcal{G}^2(P^2_{N^2})$ and repeat Round 1.

    
    This continues till all the objects are allocated to some agents. Note that this always results in a deterministic assignment.
    
\end{enumerate}

The following theorem is due to \cite{ma1994strategy}. Later on, shorter proofs are found by \cite{svensson1994queue}, \cite{anno2015short}, and \cite{sethuraman2016alternative}.
\begin{theorem*}
Suppose that $E$ is the initial deterministic endowment. Then, TTC is the unique deterministic assignment rule that satisfies Pareto-efficiency, individual rationality, and strategy-proofness. 
\end{theorem*}

\section{Results}\label{sec: results}
From \cite{ma1994strategy}'s result, we know that the TTC is the unique efficient, individually rational, and strategy-proof deterministic assignment rule when the endowments are deterministic. In this section, we extend this result to probabilistic assignments and prove that even in such a setting, TTC continues to be the only efficient, individually rational, and top-strategy-proof assignment rule starting from deterministic initial endowments.

\begin{theorem}\label{the: detend} Suppose that $E$ is the initial endowment. Then, a PR is  SD-Pareto-efficient, SD-individually rational, and SD-top-strategy-proof if and only if it is the TTC rule. 
\end{theorem}

\begin{proof}
The if part of the result follows from \cite{ma1994strategy} as the TTC rule satisfies the three properties in the deterministic setup. We proceed to show the only-if part by induction on $n$. First, consider $n=2$. If the top-ranked objects of each of the two agents is different, both of them get their top-ranked objects and the outcome is as that of TTC rule. If the top-ranked objects of both the agents is same i.e., $\dendow{1}$, by individual rationality, agent $1$ is assigned whole of $\dendow{1}$ and hence agent 2 is assigned whole of $\dendow{2}$. This again is the outcome of TTC rule.

Induction Hypothesis (IH): When there are at most $n-1$ agents and objects, an assignment rule $\rrule$ is Pareto-efficient, individually rational, and strategy-proof, then it is TTC rule.

We use IH to prove that the same statement holds for $n$ agents and objects. Let $P_N$ be the preference profile of the agents. Construct $\mathcal{G}(P_N)$ such that an edge $(i,j)$ is present if and only if $\prefof{}{}{} = \dendow{j}$. For all the singleton cycles $(i,i)$ in the graph, by individual rationality, $\rruleof{}{\dendow{i}}{P_N} = 1$ and the assignment is same as that in TTC rule. Now, let $\mathcal{C}=(i_1,\ldots,i_k,i_1)$ be a non-singleton cycle in $\mathcal{G}(P_N)$ (We use $i_{k+1}$ to denote $i_1$ and $i_0$ to denote $i_k$). Construct a new profile $P'_N$  by modifying the preferences of agents as follows: for every $i \in \mathcal{C}$, push $\dendow{i}$ to the second rank in the preference of $i$. That is, $\prefof{}{}{}P'_i\dendow{i}P'_i\ldots$ is the preference of $i$ in $P'_N$. For every $i \notin \mathcal{C}$, $P'_i = P_i$. 

It is to be noted that $\mathcal{G}(P_N) = \mathcal{G}(P'_N)$ and hence $\mathcal{C}$ remains to be a non-singleton cycle at both the profiles. 
We first prove the following claim which states that at $P'_N$, for the individual rationality to be satisfied by a probabilistic assignment rule, all the objects in $\mathcal{C}$ must be assigned only to the agents in $\mathcal{C}$.
\begin{claim}\label{cla: ir}
For any individually rational probabilistic assignment rule $\rrule$ and every $\dendow{j} \in \objects$ such that $j \in \mathcal{C}$, we have,
$$\suml{i \in \mathcal{C}}{\rruleof{i}{\dendow{j}}{P'_N}} = 1.$$
\end{claim}
\begin{proof}
This is satisfied trivially if $\mathcal{C}$ is empty. Hence, without loss of generality, assume that $|\mathcal{C}|>1$. 
Since $\mathcal{C}$ is forming a cycle, we have
\begin{align}
    \label{eq: setsincycle}
    \{\prefof{}{}{} \mid  i \in \mathcal{C}\} = \{\dendow{i} \mid i \in \mathcal{C}\}.
\end{align}
Since $\rrule$ is individually rational, for every $i \in \mathcal{C}$ we have,
\begin{align}
    \label{eq: equals1}
    \rruleof{i}{\{\prefof{}{}{},\dendow{i}\}}{P'_N} = 1
\end{align}
From \Cref{eq: equals1} we have,
\begin{align}
    \label{eq: equalsz}
    \suml{i \in \mathcal{C}}{(\rrule_{i\prefof{}{}{}}+\rrule_{i\dendow{i}})} = k
\end{align}
By applying \Cref{eq: setsincycle} here, we get,
\begin{align}
    \label{eq: equalsz2}
    \suml{\dendow{j} : i_j \in \mathcal{C}} {(\rruleof{i_{j-1}}{\dendow{j}}{P'_N}+\rruleof{i_j}{\dendow{j}}{P'_N})} = k
\end{align}
Since $\rruleof{\boldsymbol{\bigcdot}}{\dendow{j}}{P'_N}$ is a probability distribution, each of the terms on the left side must be at most $1$. That is, for any $\dendow{j}$ such that $i_j \in \mathcal{C}$, we have
\begin{align}
    \label{eq: atmost1}
    \rruleof{i_{j-1}}{\dendow{j}}{P'_N}+\rruleof{i_j}{\dendow{j}}{P'_N} \leq 1
\end{align}
This, when combined with \Cref{eq: equalsz2} and the fact that $|\mathcal{C}| = k$, gives
\begin{align}
    \label{eq: eachexact1}
     \rruleof{i_{j-1}}{\dendow{j}}{P'_N}+\rruleof{i_j}{\dendow{j}}{P'_N} = 1
\end{align}
Hence, the claim follows.
\end{proof}
We now impose Pareto-efficiency on the individually rational probabilistic assignment rule $\rrule$. Clearly, from \Cref{cla: ir}, endowments of all the agents in $\mathcal{C}$ are re-assigned to the agents in $\mathcal{C}$. As we can see, the only such probabilistic assignment that is not Pareto-dominated is to assign $\prefof{}{}{}$ to each agent $i \in \mathcal{C}$. Any other individually rational probabilistic assignment will be Pareto-dominated by this assignment. This is because, for each agent $i \in \mathcal{C}$, $\prefof{}{}{}$ is the most preferred object of agent $i$ in $\mathcal{G}(P'_N)$. Clearly, this assignment is the same as that of TTC rule. Thus, the next claim follows.
\begin{claim}\label{cla: pe}
Any Pareto-efficient and individually rational probabilistic assignment rule $\rrule$ gives only TTC assignment for all the agents in the cycle $\mathcal{C}$ at the profile $P'_N$.
\end{claim}
We now impose top-strategy-proofness on the Pareto-efficient and individually rational probabilistic assignment rule $\rrule$. We modify $P'_N$ to get $P^1_N$ as follows: change the preference of exactly one agent $i_1$ in $\mathcal{C}$ to $P_{i_1}$ and let the preferences of all the other agents be same as that in $P'_N$. We know that, $\rruleof{i_1}{\prefof{i_1}{}{}}{P'_N} = 1$. By top-strategy-proofness, $\rruleof{i_1}{\prefof{i_1}{}{}}{P^1_N} = 1$ (else, $i_1$ could have misreported its preference to be $P'_{i_1}$ instead of $P_{i_1}$ and get a better probability for $\prefof{i_1}{}{}$). From individual rationality, the agent endowed with the object $\prefof{i_1}{}{}$ should now get the object it has an outgoing edge to and so on. Thus, all the agents in $\mathcal{C}$ also continue to get the same objects as at profile $P'_N$ (that is, they get their top-ranked objects).

Now, we modify $P'_N$ to get $P^2_N$ as follows: change the preference of exactly two agents $i_1,i_2$ in $\mathcal{C}$ to $P_{i_1}$ and $P_{i_2}$ and let the preferences of all the other agents be same as that in $P'_N$. We know that, $\rruleof{i_1}{\prefof{i_1}{}{}}{P^1_N} = \rruleof{i_2}{\prefof{i_2}{}{}}{P^1_N} = 1$. By top-strategy-proofness, $\rruleof{i_1}{\prefof{i_1}{}{}}{P^2_N} = \rruleof{i_2}{\prefof{i_2}{}{}}{P^2_N} = 1$ (else, $i_1$ and $i_2$ could have unilaterally misreported their preferences to be $P'_{i_1}$ and $P'_{i_2}$ instead of $P_{i_1}$ and $P_{i_2}$ respectively and get better probabilities for their top-ranked objects). Again, by individual rationality, all the other agents in $\mathcal{C}$ also continue to get the same objects as at profile $P'_N$ (that is, they get their top-ranked objects).

Repeat this process $z$ times. Clearly, $P^z_N = P_N$. This implies that the assignments for all the agents in $\mathcal{C}$ at $P_N$ and $P'_N$ are same as that of TTC rule.

Let $\mathcal{S}$ be the set of agents forming singleton cycle in $\mathcal{G}(P_N)$. Now consider another instance of object assignment after removing $\mathcal{S} \cup \mathcal{C}$ from the set of agents and their endowments from the set of objects. From IH, the outcome in the resultant instance is same as that of TTC rule. Hence, the rule has to be TTC rule.
\end{proof}

}

	\bibliographystyle{plainnat}
	\setcitestyle{numbers}
	\bibliography{references.bib}

\begin{thebibliography}{15}
\providecommand{\natexlab}[1]{#1}
\providecommand{\url}[1]{\texttt{#1}}
\expandafter\ifx\csname urlstyle\endcsname\relax
  \providecommand{\doi}[1]{doi: #1}\else
  \providecommand{\doi}{doi: \begingroup \urlstyle{rm}\Url}\fi

\bibitem[Abdulkadiro{\u{g}}lu and S{\"o}nmez(1999)]{abdulkadirouglu1999house}
Atila Abdulkadiro{\u{g}}lu and Tayfun S{\"o}nmez.
\newblock House allocation with existing tenants.
\newblock \emph{Journal of Economic Theory}, 88\penalty0 (2):\penalty0
  233--260, 1999.

\bibitem[Anno(2015)]{anno2015short}
Hidekazu Anno.
\newblock A short proof for the characterization of the core in housing
  markets.
\newblock \emph{Economics Letters}, 126:\penalty0 66--67, 2015.

\bibitem[Athanassoglou and Sethuraman(2011)]{athanassoglou2011house}
Stergios Athanassoglou and Jay Sethuraman.
\newblock House allocation with fractional endowments.
\newblock \emph{International Journal of Game Theory}, 40:\penalty0 481--513,
  2011.

\bibitem[Aziz(2015)]{aziz2015generalizing}
Haris Aziz.
\newblock Generalizing top trading cycles for housing markets with fractional
  endowments.
\newblock \emph{arXiv preprint arXiv:1509.03915}, 2015.

\bibitem[Bogomolnaia and Heo(2012)]{bogomolnaia2012probabilistic}
Anna Bogomolnaia and Eun~Jeong Heo.
\newblock Probabilistic assignment of objects: Characterizing the serial rule.
\newblock \emph{Journal of Economic Theory}, 147\penalty0 (5):\penalty0
  2072--2082, 2012.

\bibitem[Bogomolnaia and Moulin(2001)]{bogomolnaia2001new}
Anna Bogomolnaia and Herv{\'e} Moulin.
\newblock A new solution to the random assignment problem.
\newblock \emph{Journal of Economic theory}, 100\penalty0 (2):\penalty0
  295--328, 2001.

\bibitem[Chun and Yun(2020)]{chun2020upper}
Youngsub Chun and Kiyong Yun.
\newblock Upper-contour strategy-proofness in the probabilistic assignment
  problem.
\newblock \emph{Social Choice and Welfare}, 54:\penalty0 667--687, 2020.

\bibitem[Ekici(2024)]{ekici2024pair}
{\"O}zg{\"u}n Ekici.
\newblock Pair-efficient reallocation of indivisible objects.
\newblock \emph{Theoretical Economics}, 19\penalty0 (2):\penalty0 551--564,
  2024.

\bibitem[Gibbard(1977)]{gibbard1977manipulation}
Allan Gibbard.
\newblock Manipulation of schemes that mix voting with chance.
\newblock \emph{Econometrica: Journal of the Econometric Society}, pages
  665--681, 1977.

\bibitem[Ma(1994)]{ma1994strategy}
Jinpeng Ma.
\newblock Strategy-proofness and the strict core in a market with
  indivisibilities.
\newblock \emph{International Journal of Game Theory}, 23\penalty0
  (1):\penalty0 75--83, 1994.

\bibitem[Sen(2011)]{sen2011gibbard}
Arunava Sen.
\newblock The gibbard random dictatorship theorem: a generalization and a new
  proof.
\newblock \emph{SERIEs}, 2\penalty0 (4):\penalty0 515--527, 2011.

\bibitem[Sethuraman(2016)]{sethuraman2016alternative}
Jay Sethuraman.
\newblock An alternative proof of a characterization of the ttc mechanism.
\newblock \emph{Operations Research Letters}, 44\penalty0 (1):\penalty0
  107--108, 2016.

\bibitem[Shapley and Scarf(1974)]{shapley1974cores}
Lloyd Shapley and Herbert Scarf.
\newblock On cores and indivisibility.
\newblock \emph{Journal of mathematical economics}, 1\penalty0 (1):\penalty0
  23--37, 1974.

\bibitem[Svensson(1994)]{svensson1994queue}
Lars-Gunnar Svensson.
\newblock Queue allocation of indivisible goods.
\newblock \emph{Social Choice and Welfare}, 11\penalty0 (4):\penalty0 323--330,
  1994.

\bibitem[Y{\i}lmaz(2009)]{yilmaz2009random}
{\"O}zg{\"u}r Y{\i}lmaz.
\newblock Random assignment under weak preferences.
\newblock \emph{Games and Economic Behavior}, 66\penalty0 (1):\penalty0
  546--558, 2009.

\end{thebibliography}

\appendix
%
%
%
%
%
\end{document}